\documentclass[runningheads]{llncs}

\usepackage{amssymb}
\usepackage{graphicx}
\usepackage{amsmath}
\usepackage{amsfonts}
\usepackage{algorithm}
\usepackage{algorithmicx}
\usepackage{algpseudocode}

\renewcommand{\algorithmicrequire}{\textbf{Input:}}
\renewcommand{\algorithmicensure}{\textbf{Output:}}

\begin{document}

\mainmatter

\title{Relative Convex Hull Determination \\from Convex Hulls in the Plane}

\titlerunning{Relative Convex Hull from Convex Hulls}

\author{Petra Wiederhold \and Hugo Reyes}
\authorrunning{P. Wiederhold \and H. Reyes}

\institute{Department of Automatic Control,\\ Centro de Investigaci\'on y
de Estudios Avanzados (CINVESTAV-IPN),\\ Av. I.P.N. 2508, Col. San
Pedro Zacatenco, M\'exico 07000 D.F., M\'exico,\\
pwiederhold@gmail.com, hrb87@hotmail.com}

\index{Wiederhold, Petra}
\index{Reyes, Hugo}

\maketitle

\begin{abstract}
A new algorithm for the determination of the relative convex hull in the plane of a simple polygon $A$ with respect to another simple polygon $B$ which contains $A$, is proposed. The relative convex hull is also known as geodesic convex hull, and the problem of its determination in the plane is equivalent to find the shortest curve among all Jordan curves lying in the difference set of $B$ and $A$ and encircling $A$. Algorithms sol\-ving this problem known from Computational Geometry are based on the triangulation or similar decomposition of that difference set. The algorithm presented here does not use such decomposition, but it supposes that $A$ and $B$ are given as ordered sequences of vertices. The algorithm is based on convex hull calculations of $A$ and $B$ and of smaller polygons and polylines, it produces the output list of vertices of the relative convex hull from the sequence of vertices of the convex hull of $A$.
\end{abstract}

\section{Introduction}

The relative convex hull (RCH), also called geodesic convex hull, recently has received increasing attention in Computational Geometry \cite{ToussGeodesic}, in particular related to shortest path problems which appear in a variety of applications as in robotics, industrial manufacturing, networking, or processing of geographical data \cite{ToussSeparating},\cite{Mitchell}. It was earlier defined in the context of Digital Geometry and Topology and their applications in Digital Image Analysis, where the RCH and related structures based on geodesic metrics have been proposed as approximations of digital curves and surfaces and for multi-grid convergent estimations of curve length or surface area \cite{SklanskyKibler},\cite{Lantuejoul1},\cite{Lantuejoul2},\cite{SlobStoer},\cite{SlobZatcoStoer},\cite{DigGeomBook},\cite{LiKletteBook},\cite{Biswas}, \cite{Ishaque},\cite{RobertFaugeras},\cite{YuKlette},\cite{WiedVilla}.

The convex hull of a set $S$ in the Euclidean space is obtained by filling up $S$ with all points lying on straight line segments having end points in $S$. If $S$ is contained in another set $T$, to construct the RCH of $S$ with respect to $T$, points lying on straight line segments with end points in $S$ are added whenever these segments already belong to $B$.

In the Euclidean plane and for sets $S\subset T$, the RCH of $S$ with respect to $T$, denoted by $CH_T(S)$, is obtained by allocating a tight thread around $A$ but within $B$, see Figure \ref{fig:1}a). In this paper we study the RCH for simple polygons $S$,$T$. In \cite{Ishaque}, the RCH was considered for the more general situation where $S$ is a finite point set and $T$ is a polygonal domain. A distinct definition of RCH applies to disjoint simple polygons $S,T$, then $CH_T(S)$ is the weakly simple polygon formed by the shortest closed polygonal path without self-crossings which circumscribes $S$ but excludes $T$ \cite{ToussSeparating}, see Figure \ref{fig:1}b).

Under special conditions for the polygons $S$, $T$, $S\subset T$, the RCH coincides with the Minimum Perimeter Polygon (MPP) of $S$ with respect to $T$, also called the Minimum Length Polygon (MLP), whose frontier is the shortest Jordan curve among all Jordan curves which circumscribe $S$ but are contained in $T$ \cite{DigGeomBook},\cite{SlobStoer},\cite{SlobZatcoStoer}. The MPP was first defined for polygons $S,T$ which are point set unions of cell complexes within plane mosaics modelling the digital plane where the pixels are identified with convex not necessarily uniform tiles \cite{SklanskyBlobs},\cite{SklanskyKibler},\cite{SklanskyMeasConcav},\cite{SklanskyMPPSilh}, see Figure \ref{fig:1}c). These polygons $S,T$ are constructed as the Inner and Outer Jordan digitization of a subset of the Euclidean plane which is the interior of a given Jordan curve $\gamma $. For the digital plane modeled by the standard quadratic complex where all pixels are grid squares of the same size, $S,T$ are isothetic simple polygons and $(T\setminus S)$ is a union of grid squares called grid continuum, see Figure \ref{fig:1}d). In this case, the length of the frontier of the RCH is a multi-grid convergent estimator of the length of the Jordan curve $\gamma $ \cite{DigGeomBook},\cite{SlobStoer},\cite{SlobZatcoStoer}. Several efficient MLP algorithms are known, for example the corrected version of \cite{KletteKovYip} in \cite{LiKletteBook}, and \cite{Provencal}, but these can be applied only to digital continua or polyominoes.

\begin{figure}
\centering
\includegraphics[height=2.3cm]{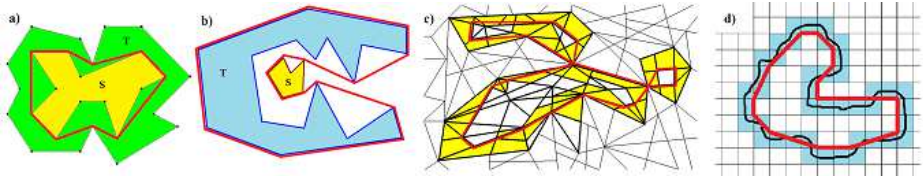}
\caption{a) RCH of a set with respect to a superset, b) RCH for two disjoint sets, c) MPP of a subcomplex of a mosaic, d) MLP of a grid continuum being a digital model of a Jordan curve.}
\label{fig:1}
\end{figure}

In this paper we propose a novel algorithm for the determination of the ordered list of all vertices of the RCH, for the general situation of given simple plane polygons $A$, $B$ such that $A\subset B$. The algorithm does not use previous triangulation or similar decompositions. Each input polygon is given as ordered set of its vertices. Our algorithm adopts some basic ideas of the algorithm published in \cite{GiselaDCGI} but presents essential corrections and improvements. A preliminary version of our algorithm was developed in \cite{Hugo}.

\section{Preliminaries}
\label{sect:Preliminaries}

Recall that a non-empty set $S\subset\mathbb{R}^2$ is \textit{convex} if for any $p,q\in S$, the straight line segment $\overline{pq}$
is contained in $S$, where $\overline{pq}$ is the set of all points $r=\lambda_1 p + \lambda_2 q$ such that $\lambda_1, \lambda_2\in \mathbb{R}$, $\lambda_1,\lambda_2\geq 0$, $\lambda_1 + \lambda_2 =1$. The \textit{convex hull} of $S$ denoted by $CH(S)$, is the intersection of all convex sets which contain $S$. Equivalently, $CH(S)$ is the set of all points which belong to straight line segments with end points in $S$. For basic topological notions we refer to \cite{Munkres}, we will denote the (topological) interior of $S$ by $int(S)$ and its frontier by $fr(S)$. A non-convex set is distinct from its convex hull via the presence of holes or cavities: Any bounded connected component of $(\mathbb{R}^2\setminus S)$ is a \textit{hole} of $S$. The closure of any connected component of $(CH(S)\setminus S)$ which is not a hole of $S$, is a \textit{cavity} of $S$.

\begin{figure}
\centering
\includegraphics[height=2.8cm]{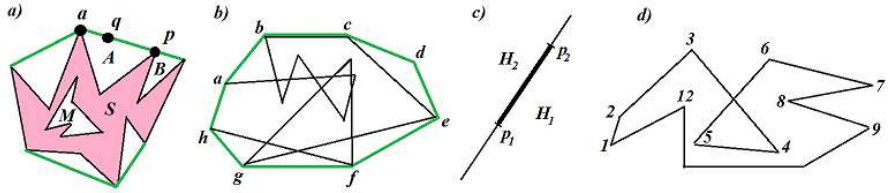}
\caption{a) $M$ is a hole of $S$, the cavities $A,B$ are distinct although they share the point $p$. The straight line segment $\overline{aq}$  is not a cover of the cavity $A$ although it belongs to $(fr(A)\setminus fr(S))$, $\overline{ap}$ is the cover of $A$. b) A polyline and its convex hull given by the vertex sequence $\langle a,b,c,d,e,f,g,h\rangle $. c) Right and left halfplanes determined by $\protect\overrightarrow{p_1p_2}$. d) Points 1,2,5,6,7,9 are examples of convex vertices (right turns), points 8 and 12 are concave vertices (left turns) of the closed polyline traced in clockwise sense.} \label{fig:2}
\end{figure}

A \textit{curve} $\gamma = \{ f(s)= (x(s), y(s))\in\mathbb{R}^2 :\  s\in [0,1] \} $ ($f:[0,1]\rightarrow \mathbb{R}^2$ continuous), is closed if $f(0)= f(1)$, simple if for any $s,t\in [0,1]$ such that $0\leq s<t<1$ it follows $f(s)\neq f(t)$; $\gamma $ is a \textit{Jordan curve} if it is simple and closed. A Jordan curve $\gamma $ separates the plane into two uniquely defined open disjoint regions: the \textit{interior of the Jordan curve} is bounded and encircled by $\gamma $, and the exterior of the Jordan curve is not bounded \cite{Munkres}. A curve is named \textit{polyline} if there exists a finite sequence of points $\{ s_0, s_1, s_2, \cdots , s_k\} $, with $0= s_0< s_1< s_2< \cdots < s_k=1$ such that all curve segments $\{ f(s): s_i\leq s\leq s_{i+1}\} $ ($i=0,1, \cdots ,k-1$) are straight line segments. The points $\{ s_0, s_1, s_2, \cdots , s_k\} $ are named vertices whenever no three consecutive points are collinear. A polyline is uniquely determined by the sequence of its vertices. A closed polyline corresponds to a closed curve, a simple polyline is a simple curve. A vertex $p$ of a polyline $\gamma $ is called \textit{extreme vertex} if its $x$-coordinate is extreme (that is, maximal or minimal) among the $x$-coordinates of all vertices of $\gamma $ or, if its $y$-coordinate is extreme among all $y$-coordinates of vertices of $\gamma $. Any extreme vertex of a polyline $\gamma $ is a vertex of the convex hull $CH(\gamma )$. A \textit{simple polygon} is defined as any non-empty bounded closed set $P\subset\mathbb{R}^2$ whose frontier forms a simple closed polyline. Hence the frontier of a simple polygon is a Jordan curve and can be represented by the finite cyclic sequence of its vertices. The convex hull of a simple polygon coincides with the convex hull of the finite set of its vertices. A simple polygon does not have holes, therefore it is non-convex if and only if it has at least one cavity. For any non-convex simple polygon $S$ in the plane and any cavity $M$ of $S$, define the \textbf{cover} of $M$ as straight line segment of maximal length belonging to $fr(M)\setminus S$. The requirement of maximal length guarantees that the cover for each cavity $M$ is unique, see Figure \ref{fig:2}a). For any ordered triple of points $p_1=(x_1,y_1)$, $p_2=(x_2,y_2)$, $p_3=(x_3,y_3)$ in the plane, its orientation is characterized by the sign of the determinant
$D(p_1, p_2, p_3) = x_1y_2 + y_1x_3 + x_2y_3 - ( x_3y_2 + x_2y_1 + x_1y_3)$. The oriented line segment $\overrightarrow{p_1p_2}$ defines an oriented line which separates $\mathbb{R}^2$ into a right halfplane $H_1$ and a left halfplane $H_2$, see Figure \ref{fig:2}c). $(p_1,p_2, p_3)$ forms a \textit{right turn} if $p_3\in H_1$, $(p_1,p_2, p_3)$ forms a \textit{left turn} if $p_3\in H_2$. Using the standard cartesian coordinate system in the plane, for a closed (simple) polyline $L$ traced in clockwise sense, see Figure \ref{fig:2}d), for any three consecutive vertices $p_1, p_2, p_3$ of $L$ we have the following: $(p_1, p_2, p_3)$ forms a right turn if and only if $D(p_1, p_2, p_3)<0$; then $p_2$ is called a \textit{convex vertex}; $(p_1, p_2, p_3)$ forms a left turn if and only if $D(p_1, p_2, p_3)>0$, then $p_2$ is called a \textit{concave vertex}. $p_1, p_2, p_3$ are collinear points if and only if $D(p_1, p_2, p_3)=0$.

\section{Definition and Properties of the Relative Convex Hull}

\begin{definition}
\label{def:RCH}
Let $A,B\subset \mathbb{R}^n$ be non-empty sets such that $A\subseteq B$. Then $A$ is called \textbf{$B$-convex} if any straight line segment lying in $B$ whose both end points belong to $A$, is contained in $A$. The \textbf{relative convex hull of $A$ with respect to $B$}, denoted by $CH_B(A)$, is defined as the intersection of all $B$-convex sets which contain $A$.
\end{definition}

It is evident that each set $A$ is $A$-convex, and that if $A$ is convex and $A\subset B$ then $A$ also is $B$-convex. The following properties can be derived from the definitions of $CH(A)$ and $CH_B(A)$:

\begin{lemma}
\label{lemma:RCHproperties}
\noindent (i) $A\subset CH_B(A)\subset B$, $B$ is the largest $B$-convex set which contains $A$ whereas $CH_B(A)$ is the smallest such set.

\noindent (ii) $CH_B(A)\subset CH(A)$.

\noindent (iii) $A$ is convex if and only if $CH_B(A)= CH(A) =A$.

\noindent (iv) $CH(A)\subset B$ if and only if $CH_B(A) = CH(A)$.

\noindent (v) If $B$ is convex then $CH_B(A) = CH(A)$.
\end{lemma}

\begin{proof}
The definitions and constructions of $CH(A)$ and $CH_B(A)$ imply (i) and (ii); (iii) follows from (ii) and since $A$ is convex if and only if $CH(A)=A$.

\noindent (iv) Suppose $CH(A)\subset B$. Because of (ii), only $CH(A)\subset CH_B(A) $ remains to be proved. Let $p\in CH(A)$ and $M\subset B$ be any $B$-convex set containing $A$. We have to prove that $p\in M$. For $p\in A$ this is trivial, so assume $p\not\in A$. Since $CH(A)$ is the set of all straight line segments having end points in $A$, $p$ belongs to some straight line segment with end points $a,b\in A$. But then $a,b$ belong also to $M\subset B$. The segment $\overline{ab}$ is contained in $CH(A)$ and hence, by the hypothesis, to $B$. Since $M$ es $B$-convex, $p\in \overline{ab}\subset M$ which completes the proof of $CH_B(A) = CH(A)$. On the other hand, $CH_B(A) = CH(A)$ means in particular that $CH(A)$ is contained in each $B$-convex set which contains $A$, but $B$ is such a set, implying $CH(A)\subset B$.

\noindent (v) $A\subset B$ with $B$ convex implies $CH(A)\subset CH(B) =B$, then (iii) gives the result.

\rightline{$\square$}
\end{proof}

As a corollary, it can be proved that a necessary condition for $CH_B(A)\neq CH(A)$ is that some concave vertex of $B$ lies in the interior of a cavity of $A$. In this paper we study the RCH only for simple polygons $A$ and $B$ in the plane, $A\subset B\subset \mathbb{R}^2$. The following properties are important for the determination of the RCH:

\begin{theorem}
\label{theorem:RCH vertex candidates}
Let $A$, $B$ be simple polygons such that $A\subset int(B)$.

\noindent (i) $CH_B(A)$ exists and is a uniquely defined simple polygon.

\noindent (ii) The frontier of the polygon $CH_B(A)$ is the Jordan curve which among all Jordan curves circumscribing $A$ and lying in $B$, has the shortest length.

\noindent (iii) Each convex vertex of $CH_B(A)$ is a convex vertex of $A$, and each concave vertex of $CH_B(A)$ is a concave vertex of $B$.

\end{theorem}

This was given by Theorem 3 from \cite{SlobStoer} and Theorem 4.6 from \cite{SlobZatcoStoer}. When the condition is weakened to $A\subset B$ then the polygon $CH_B(A)$ is simple or weakly simple, that means, its frontier can touch itself but does not cross itself, and the other properties are still valid \cite{ToussGeodesic}.

\begin{theorem}
\label{theorem:All-vertices-CH(A)}
For simple polygons $A$, $B$ such that $A\subset B$, all vertices of $CH(A)$ are vertices of $CH_B(A)$.
\end{theorem}

\begin{proof}
Any vertex of $CH(A)$ belongs to $A\subset CH_B(A)$. To prove that any vertex of $CH(A)$ is a vertex of $CH_B(A)$, we apply Lemma \ref{lemma:RCHproperties}(i) and the well-known fact that any convex simple polygon is a finite intersection of halfplanes which are determined by the straight lines generated by the polygon edges. The convex simple polygon $CH(A)$ has $k\geq 3$ vertices $a_1, a_2, \cdots , a_k$, where no three consecutive points are collinear, and $k$ edges $\overline{a_1a_2}$, $\overline{a_2a_3}$, $\cdots $,  $\overline{a_ka_1}$. Supposing a clockwise tracing of the Jordan curve $fr(CH(A))$, let $H_i$ be the right halfplane  of the oriented straight line generated by the line segment $\overrightarrow{a_ia_{i+1}}$ for $i=1, 2, \cdots k-1$, and $H_k$ be the right halfplane of $\overrightarrow{a_ka_1}$. Then $CH(A)= H_1 \cap H_2 \cap \cdots \cap H_k$, and all these halfplanes are pairwise distinct. For any vertex $a_i$ of $CH(A)$, $a_{i-1}, a_i, a_{i+1}$ belong to $A\subset CH_B(A)\subset CH(A)\subset H_{i-1} \cap H_i \cap H_{i+1}$. This implies that $a_{i-1}, a_i, a_{i+1}\in fr(CH_B(A))$ and that $CH_B(A)$ cannot contain elements of the straight line generated by the segment $\overline{a_i a_{i+1}}$ but lying outside this segment. In consequence, in particular $a_i$ is a vertex of $CH_B(A)$. Note that the argument of our proof is independent of a discussion weather $\overline{a_i a_{i+1}}$ belongs to $B$ or not.

\rightline{$\square$}
\end{proof}

The last theorem was briefly mentioned on p.126 of \cite{SlobZatcoStoer} without proof, and it was stated in \cite{GiselaDCGI} with a wrong proof.

\section{Previous Algorithms of Determining the Relative Convex Hull for Simple Polygons in the Plane}
\label{sect:known algs}

G. Toussaint proposed in \cite{ToussOptAlg},\cite{ToussGeodesic} to transform the problem of determining $CH_B(A)$ into the problem of finding the shortest path between two vertices of a new simple polygon which first is triangulated. That algorithm has linear time complexity in terms of the total number $k$ of vertices of $A$ and $B$, but it makes essential use of the triangulation of $M$ which can be achieved by a complicated process in $\mathcal{O}(k\, log(log(k)))$ time. In \cite{SlobStoer} and \cite{SlobZatcoStoer} several ideas for algorithms to determine the RCH were suggested, which are based on decompositions of the polygons such as trapezoidation or partition into pseudomonotone polygons. In the context of digital curve analysis, some algorithms not based on triangulations for calculating the MPP or MLP are known, for example \cite{SklanskyKibler},\cite{SklanskyMPPSilh},\cite{KletteKovYip},\cite{DigGeomBook},\cite{LiKletteBook},\cite{Provencal}, but these algorithms solve the RCH problem only for special difference sets $(B\setminus int(A))$ such as grid continua or polyominoes or special cell complexes.

The algorithm published in \cite{GiselaIEEE},\cite{GiselaDCGI},\cite{GiselaJournal},\cite{LiKletteBook} starts with calculating the convex hulls of $A$ and $B$. The list of vertices of $CH(A)$ is completed by inserting vertices from cavities of $B$ until the output list of all vertices of $CH_B(A)$ is obtained. The construction of the output list follows a recursive process which searches for intersections of cavities of $A$ and $B$. Whenever such intersection is detected, a new outer polygon $O$ and a new inner polygon $I$ are formed, and the problem of finding $CH_O(I)$ is treated to obtain missing RCH vertices of $CH_B(A)$. Subsequently, the recursive process works in each step with smaller newly generated outer and inner polygons and calculates their convex hulls. The author affirms that after sufficiently many recursion steps, the base case of the recursion is achieved where the new inner polygon is a triangle. The idea of such a recursive process was first suggested by two theorems on the shortest path between two vertices of a polygon and a series of drawings on p.122-124 in \cite{SlobZatcoStoer} where the explanation was not detailed at all. In certain situations, the algorithm from \cite{GiselaDCGI} does not produce the correct result of all vertices of $CH_B(A)$. The reason for this lies in the geometric nature of the RCH problem for general simple polygons $A, B$ which was oversimplified in \cite{GiselaDCGI}; its recursion is theoretically not justified. The new polygon $I$ sometimes is not contained in $O$ or is not a simple polygon. For finding the missing vertices, additional regions have to be investigated in each step. It is also possible that the process stops when $I$ becomes convex but is not a triangle.

\section{A New Algorithm of Determining the Relative Convex Hull for Simple Polygons in the Plane}
\label{sect:new alg}

\subsection{Vertex Lists, Convex Hull Determination and Cavity Detection}

The new algorithm will be explained with the help of the example shown in Figure \ref{fig:ExamplePolygons}. The input data consist of two simple polygons $A$, $B$ satisfying $A\subset B$, given as ordered sequence of vertices: $A=\langle p_1, p_2, \dots, p_n \rangle $, $B=\langle q_1, q_2, \dots, q_m \rangle $ representing the frontier of each polygon due to the clockwise tracing. We suppose $p_1$ as an extreme vertex of $A$, $q_1$ extreme for $B$ which can be achieved by a simple pre-processing of both lists. Hence $p_1$ is a vertex of $CH(A)$ and hence of $CH_B(A)$, by theorem \ref{theorem:All-vertices-CH(A)}. The algorithm produces an ordered list of all vertices of $CH_B(A)$ as output data, starting with $p_1$ and corresponding to a clockwise tracing of the frontier of $CH_B(A)$.

\begin{figure}
\centering
\includegraphics[height=4cm]{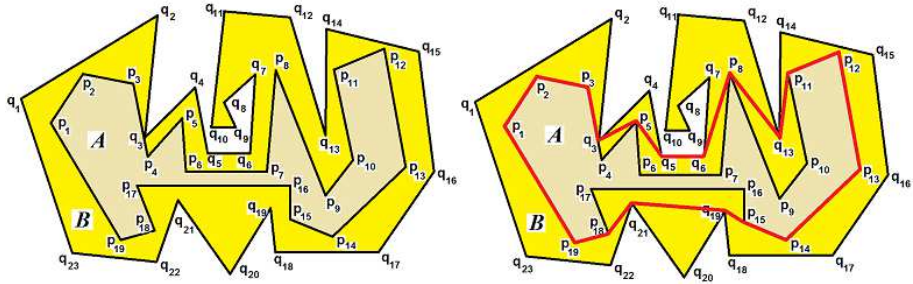}
\caption{Left: example of input data given by an inner polygon $A$ and an outer polygon $B$. Right: the sides of $CH_B(A)$ are marked by heavy red lines.}
\label{fig:ExamplePolygons}
\end{figure}

Our algorithm starts with determining all vertices of the convex hulls of both $A,B$ which are stored in the vertex lists $CH(A), CH(B)$, respecting the clockwise tracing. This can be done for example by the efficient Melkman-Algorithm \cite{Melkman}. As a particularity of this algorithm, the last vertex which was confirmed as vertex of the convex hull and hence appears at the end of the output list, is repeated in that list as first point, we eliminate this first point from the list. So we obtain the vertex list $CH(A)$ starting with $p_1$ and containing a selection of points from the list $A$ whose original ordering and internal indices are preserved, similarly for $CH(B)$ starting with $q_1$. The vertex list $CH(A)$ is considered as initial output list of the vertices of the RCH. By subsequent steps of our algorithm, all other RCH vertices are found and inserted into this list $CH(A)$ at appropriate positions. Therefore, the format of a double ended queue owned by the vertex list $CH(A)$ as output of the Melkman-Algorithm, cannot be preserved during subsequent steps of our method. We apply later again the Melkman-Algorithm \cite{Melkman} which produces the vertex list of the convex hull for any input vertex list of a polyline not necessarily closed or forming a simple polygon, and it always respects the order in the input vertex list.

In each vertex list $A$, $B$, $CH(A)$, $CH(B)$, we copy its first point as added at the list end but having a new index. This permits to study all sides of each polygon, including the line segment connecting the last vertex with the first one, without producing errors in the indices when performing our algorithm. For our example, this produces $CH(A) = (p_1, p_2, p_{12}, p_{13}, p_{14}, p_{19}, p_{20}=p_1)$, $CH(B) = (q_1, q_2, q_{11}, q_{12}, q_{14}, q_{15}, q_{16}, q_{17}, q_{20}, q_{23}, q_{24}=q_1)$.

Since each point of the vertex list $CH(A)$, besides having an $CH(A)$-index $i$, also preserves its original index from the vertex list $A$, a cavity of the polygon $A$ is easily detected during tracing the list $CH(A)$: When consecutive vertices have a difference strictly mayor than $1$ between their own indices, $CH(A)_i=p_k$, $CH(A)_{i+1}=p_l$, and $|k-l|\geq 2$, then $A$ has a cavity whose cover is given by the line segment $\overline{p_kp_l}$. Cavities of $B$ can be detected in the same manner from the list $CH(B)$. This idea was adopted from \cite{GiselaDCGI}. In our example, $i=2$ indicates that $\overline{p_2p_{12}}$ is the cover of a first cavity of $A$.

\subsection{Processing of one Cavity}

As in \cite{GiselaDCGI}, whenever a cavity of $A$ is found, it is considered as a \textit{new polygon $O$} determined by its vertices $\langle CH(A)_i=p_k, p_{k+1}, p_{k+2}, \cdots ,$ $p_{k+r}= p_{l}= CH(A)_{i+1}\rangle $ for some $r\geq 1$ which always is a simple closed polyline in counterclockwise order. For our example, $i=2$, $O= \langle p_2, p_3, p_4, \cdots , p_{11}, p_{12}\rangle $.

The next step is to construct a \textit{new polyline $I$} whose convex hull, if it has at least three vertices, provides vertices of $B$ which are vertices of the RCH and should be inserted in the list $CH(A)$ between $CH(A)_{i}$ and $CH(A)_{i+1}$. Let $I$ be the sequence starting with $CH(A)_{i+1}$, $CH(A)_i$ and then containing all vertices from $B$, in the same order as in $B$, which belong to the set $(O \setminus \overline{CH(A)_iCH(A)_{i+1}})$ which is the polygon $O$ with exception of its cover $\overline{CH(A)_iCH(A)_{i+1}}$. Only in the case that all those vertices selected from $B$ are vertices of the same cavity of $B$, our definition of $I$ coincides with that of \cite{GiselaDCGI}. For our example, $I= \langle p_{12}, p_{2}, q_{3}, q_{4}, q_{5}, q_{6}, q_{7}, q_{8}, q_{9}, q_{10}, q_{13} \rangle $ represents a closed polyline in counter-clockwise sense, but it does not form a simple polygon, and the curve is not completely contained in $O$. All points $q_k$ of $I$ with exception of $q_{13}$ belong to the same cavity of $B$.

The Melkman-Algorithm \cite{Melkman} is applied to determine the convex hull of $I$. In our example, this produces the output $CH(I)= \langle q_{13}, p_{12}, p_{2}, q_{3}, q_{5}, q_{6}, q_{13}\rangle $. After eliminating the first point which is repeated and the end points of the cover which already belong to $CH(A)$, we obtain the following new vertices which will be inserted into the list $CH(A)$: $q_{3}$, $q_{5}$, $q_{6}$, $q_{13}$. The updated list $CH(A)$ then contains vertices both from $A$, $B$: $CH(A)= (p_1, p_2, q_{3}, q_{5}, q_{6}, q_{13}, p_{12}, p_{13}, p_{14}, p_{19}, p_{20})$. This current list $CH(A)$ represents two special line segments, each one connecting a vertex from $A$ with a vertex from $B$: $\overline{p_2q_{3}}$ and $\overline{q_{13}p_{12}}$. We will use these segments to form polylines whose convex hulls will provide eventually missing vertices of the RCH. These polylines were not defined or used in the algorithm of \cite{GiselaDCGI}.

\begin{definition}
Let $b_1, b_2, \cdots , b_k$ be the vertices of $CH(I)$ which were inserted into $CH(A)$ at the index $i$ due to the procedure described above in order to generate the current list
\[
CH(A)= (CH(A)_{1}, CH(A)_{2}, \cdots , CH(A)_{i}, b_1, b_2, \cdots , b_k, CH(A)_{i+k+1}, \cdots )\  .
\]
Define a \textbf{starting $O$-polygon} $O_S$ by the vertex sequence starting with $CH(A)_{i}$, $CH(A)_{i+1}$ and then containing all vertices which in the vertex list $B$ are previous to $CH(A)_{i+1}=b_1$, copying them in reversed order, until the first vertex which lies outside $O$. Let $I_S$ be the polyline starting with $CH(A)_{i+1}$, $CH(A)_{i}$ and then containing all vertices from the vertex list $A$, copying their ordering, which belong to $(O_S\setminus \overline{CH(A)_{i}b_1})$.

\noindent Similarly, define an \textbf{ending $O$-polygon} $O_E$ by the vertex sequence starting with $CH(A)_{i+k+1}$, $CH(A)_{i+k}=b_k$ and then containing all vertices which in the vertex list $B$ are subsequent to $CH(A)_{i+k}=b_k$, copying their ordering, until the first vertex which lies outside $O$. Let $I_E$ be the polyline starting with $CH(A)_{i+k+1}$, $CH(A)_{i+k}=b_k$ and then containing all vertices from the vertex list $A$, copying their ordering, which belong to $(O_E\setminus \overline{b_k CH(A)_{i+k+1}})$.
\end{definition}

By this definition, $O_S$ is generated in counter-clockwise sense whereas $I_S$, $O_E$ and $I_E$ are polylines traced in clockwise sense.

\begin{lemma}
All vertices of $CH(I)$, $CH(I_S)$, $CH(I_E)$ are vertices of $CH_B(A)$.
\end{lemma}

\textit{Idea of Proof:} Let $O$ be a cavity of $A$ with cover $\overline{pq}$ and at least one vertex of $B$ inside $O\setminus \overline{pq}$. $O$ is a simple polygon. Due to Theorem \ref{theorem:RCH vertex candidates}(ii), all vertices of $CH_B(A)$ belonging to $R(O)$ are vertices of the shortest polygonal Jordan path which circumscribes $A$ but lies in $B$. As consequence, the polygonal subpath from $p$ to $q$  is the shortest path between $p,q$ as vertices of the weakly simply polygon $O\cap B$. By Theorem 4.4 of \cite{SlobZatcoStoer} (whose validity has to be generalized from a simple to a weakly simple polygon), this subpath is contained in $CH(I)$. Together with the fact that all vertices and edges of $CH_B(A)$ cannot intersect $int(A)$, it can be proved that each vertex of $CH(I)$ is a vertex of $CH_B(A)$. The polygons $O_S$, $O_E$, $I_S$, $I_E$ are simple and $I_S\subset O_S$, $I_E\subset O_E$. The subpath of $fr(CH_B(A))$ from $CH(A)_{i}$ to $b_1$ passing through certain vertices of $A$ (if any), is the shortest path between these vertices of the simple polygon $O_S\cap B$, it also belongs to $fr(CH_{O_S}(I_S))$. By Theorem \ref{theorem:All-vertices-CH(A)}, all vertices of $CH(I_S)$ are vertices of $CH_B(A)$; similarly for $I_E$. \textit{(End of Idea of Proof)}

The Melkman-Algorithm \cite{Melkman} is applied for calculating the lists $CH(I_S)$, $CH(I_E)$, which after eliminating the points which are repeated or already belonging to the list $CH(A)$, have to be inserted into the list $CH(A)$: new points provided by $CH(I_S)$ are inserted between $CH(A)_{i}$ and $CH(A)_{i+1}=b_1$, new points from $CH(I_E)$ are inserted between $CH(A)_{i+k}=b_k$ and $CH(A)_{i+k+1}$. In our example, $i=2$, $CH(A)_{i}=p_2$, $CH(A)_{i+1}=b_1=q_3$, $k=4$, $CH(A)_{i+k}=b_k=q_{13}$, $CH(A)_{i+k+1}=p_{12}$, $O_S=\langle p_2, q_3, q_2\rangle $, $I_S=(q_3, p_2, p_3)$ is convex and provides the new point $p_3$ to be inserted between $p_2$ and $q_3$. $O_E=\langle p_{12}, q_{13}, q_{14}\rangle $, $I_E=(p_{12}, q_{13}, p_{11})$ is convex, so that only $p_{11}$ has to be inserted between $q_{13}$ and $p_{12}$. The new list is $CH(A)= (p_1, p_2, p_3, q_{3}, q_{5}, q_{6}, q_{13}, p_{11}, p_{12}, p_{13}, p_{14}, p_{19}, p_{20})$. This completes to process the cavity of $A$ starting at the vertex with $CH(A)$-index $i$. Note that during the whole procedure just described, this starting index $i$ is not changed and points are inserted only after that index. Comparing the current list $CH(A)$ with Figure \ref{fig:ExamplePolygons} we see that within the actual cavity, more RCH vertices have to been detected, but the list $CH(A)$ will guide us naturally to discover these missing points.

\subsection{Detection and Processing of Subsequent Cavities}

The algorithm continues tracing the vertex list $CH(A)$ which has been updated by processing the cavity previously detected, increasing the $CH(A)$-index $i$ and looking for consecutive vertices whose own indices have a difference more than $1$. This test is done only for consecutive vertices which both are from $A$, or both from $B$. When two points are consecutive in $CH(A)$ but one is from $A$ and the other from $B$, then the point from $B$ was inserted as result of the treatment of the special polygons $O_S$ or $O_E$, and no more vertices of the RCH are missing between these two points.

Whenever in the list $CH(A)$ two consecutive points of $A$, $CH(A)_i=p_k$ and $CH(A)_{i+1}=p_l$, such that $|k-l|\geq 2$, are found, then $p_kp_l$ covers some kind of ``cavity" of $A$ and the whole ``Processing of one cavity" described in the previous section, is performed. This includes the analysis of the polygons and polylines $O$, $I$, $O_S$, $I_S$, $O_E$, $I_E$, resulting in an updated vertex list $CH(A)$. The same is done when such two consecutive points of $B$, $CH(A)_i=q_k$ and $CH(A)_{i+1}=q_l$, are detected, but then the ``Processing of one cavity" is applied with the roles of $A$ and $B$ interchanged (points $q_j$ instead of $p_j$ and vice versa).

In our example, the next such situation is found for $i=4$ and points of $B$: $CH(A)_4=q_3, CH(A)=q_5$. Following faithfully the procedure with roles of $A$ and $B$ interchanged, we obtain $O=\langle q_3, q_4, q_5\rangle $ which is a cavity of a cavity of $B$ with one vertex of $A$ inside, giving $I=\langle q_5, q_3, p_5\rangle $. $I$ is convex and provides only the new point $p_5$. The special segments $\overline{q_3p_5}$ and $\overline{p_5q_5}$ generate $O_S= \langle q_3, p_5, p_4\rangle $ and $O_E= \langle q_5, p_5, p_6\rangle $ which both do not contain vertices of $B$, hence $I_S= \langle p_5, q_3\rangle $ and $I_E= \langle q_5, p_5\rangle $ are degenerated to line segments and do not provide more points to be inserted into the vertex list. We obtain as current list $CH(A) = (p_1, p_2, p_3, q_{3}, p_5, q_{5}, CH(A)_7=q_{6}, q_{13}, p_{11}, p_{12}, p_{13}, p_{14}, p_{19}, p_{20})$.

The next jump in the indices is found at $i=7$ again with points from $B$: $CH(A)_7=q_{6}, CH(A)_7=q_{13}$. We should be careful using geometrical concepts, the segment $\overline{q_{6}q_{13}}$ covers some part of $B$ which is neither a cavity nor a cavity of a cavity of $B$. We obtain $O= \langle q_6, q_7, q_8, q_9, q_{10}, q_{11}, q_{12}, q_{13}\rangle $, $I= \langle q_{13}, q_6, p_8 \rangle $ which is convex and provides only the new point $p_8$. $O_S, O_E$ are not interesting since $I_S, I_E$ degenerate to line segments and do not provide more points:
$CH(A)= (p_1, p_2, p_3, q_{3}, p_5, q_{5}, q_{6}, p_8, q_{13}, p_{11}, p_{12}, p_{13}, p_{14}, p_{19}, p_{20})$. The next pair of points to be treated is found as $p_{14}, p_{19}$, where the polylines $O, I$ provide the new RCH vertices $q_{19}$ and $q_{21}$, and then we need $O_S, I_S$ to discover $p_{15}$ and also $O_E, I_E$ to detect the last RCH vertex $p_{18}$ which completes the correct determination of the RCH shown in Figure \ref{fig:ExamplePolygons}.

\subsection{Pseudocode, Implementation, and Complexity}

Figure \ref{Bsp:Hugo} shows an example where the RCH was calculated by our algorithm implemented in Matlab R2012a. The example was designed in \cite{Hugo} to contain several interesting situations, such as a convex cavity of $A$ with vertices of $B$ inside, a non-convex cavity of $A$ with vertices of $B$ inside, vertices of $A$ inside interesting parts of $B$, a part of $fr(B)$ collinear with the frontier of a cavity of $A$. In the left part of each figure, both polygons $A, B$ are isothetic and the difference set $(B\setminus A)$ looks like a grid continuum, such that in this part we apply our algorithm to solve the MLP problem. The figure shows that the RCH problem, even for the MLP case, cannot be solved by the recursion of \cite{GiselaDCGI}.

Figures \ref{CodeMain}, \ref{CodeProcedure} present a pseudocode of our algorithm which is not yet optimized. To estimate the time complexity of our method, suppose that the input polygons $A$ and $B$ have $n$ and $m$ vertices, respectively. Not only the Melkman-Algorithm is applied and computes the convex hull of any polyline given as ordered sequence of $k$ vertices in linear time $\mathcal{O} (k)$. In several steps, our method needs to decide whether a point belongs to the right or left halfplane of a straight line segment, where the determinant described in Section \ref{sect:Preliminaries} is used. Also it has to be determined whether a point lies inside or outside a simple polygon given by its vertex list. When this list corresponds to a clockwise order tracing, then a point is inside the polygon if it belongs to the right halfplanes of all polygon edges. Such verifications are needed in our algorithm for polygons given by small subsets of vertices of $A, B$, so that their time complexity can be considered as linear in dependance of $m+n$.

Up to three distinct convex hulls have to be computed for each ``cavity" intersection of $A$ and $B$. $A$ has a maximum number of $\lfloor n/2\rfloor$ cavities. Each such cavity of $A$ could have vertices of $B$ inside. These vertices belong to the set of concave vertices of $B$ which could have almost $m$ elements. This gives a quadratic time complexity in the worst case. Another problem is the possible existence of interleaved and interlaced cavities within other cavities. Although our algorithm is not recursive but iterative, each cavity lying inside another cavity, when not treated immediately, is detected later when tracing the updated vertex list $CH(A)$ and then treated. So, as also observed in \cite{GiselaDCGI}, only in cases when the ``deepness" of such ``stacked cavities" is bounded by some constant and the cavities in general are ``well distributed" then our algorithm can present a nearly linear time complexity behaviour.

\section{Conclusion and Future Work}
This paper proposes an algorithm for the determination of the list of all vertices of the relative convex hull, for the general situation of given simple plane polygons $A,B$ such that $A\subset B$. This algorithm does not use triangulation or similar decompositions of the difference set between $B$ and $A$ as preprocessing. The ordered input vertex sequences of $A$ and $B$ are processed going forward to generate the output list of vertices of $CH_B(A)$ by inserting points iteratively into the list of vertices of the convex hull of $A$.

Near future work previews to complete the formal proof of correctedness of our algorithm and the solution of some pendent details such as the insertion of the new vertices found from the convex hull of the polyline $I$ into the current vertex list $CH(A)$ in the ``correct" order, or the treatment of the presence of collinear (non-consecutive) vertices of $A$ or $B$, a situation which interestingly is forbidden for algorithms based on triangulation \cite{ToussGeodesic}.

\begin{figure}
\centering
\includegraphics[height=7.4cm]{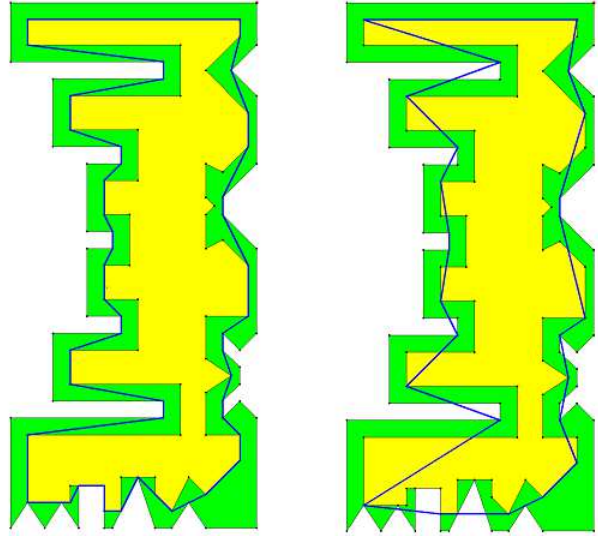}
\caption{The relative convex hull determined by the new algorithm implemented in Matlab, and by the algorithm of \cite{GiselaDCGI}, for an example developed in \cite{Hugo}.}
\label{Bsp:Hugo}
\end{figure}

\vskip0.2cm
\noindent{\bf Acknowledgement:}
The first author gratefully acknowledges support for this research from SEP and CONACYT Mexico, grant No. CB-2011-01-166223. The authors would like to thank very much to the reviewers for their careful study of the work, and for their constructive criticism and helpful comments which were important to improve the presentation of the paper.

\begin{figure}
\begin{algorithm}[H]
\begin{algorithmic}[1]

\Require Simple polygons $A$, $B$ with $A \subset B$ given by vertex lists $A=\langle p_1, p_2, \dots, p_n
\rangle $, $B=\langle q_1, q_2, \dots, q_m \rangle $ (clockwise traced).

\Ensure List of vertices of the relative hull $CH_B(A)$ stored in the actualized list $CH(A)$ (clockwise traced).

\hrule
\Statex
\Statex

\State Initialize $CH(A)=\emptyset$, $i=1$.

\State Determine the vertices of convex hull of $A$ by the Melkman algorithm
stored in the list $CH(A)$ which has $s$ elements.

\State Eliminate the first element of $CH(A)$.

\State Extend the lists $A$, $B$ y $CH(A)$ adding at the end a copy of the first element with new index.

\While{$i<s$}

\If {a cavity is detected between $CH(A)_i$ and $CH(A)_{i+1}$,}

\If {The cavity is between points of $A$,}

\State CAV$(CH(A),i,s,A,B)$ \Else \State CAV$(CH(A),i,s,B,A)$

\EndIf

\EndIf

\State i=i+1

\EndWhile

\State The actualized list $CH(A)$ contains all vertices of the relative convex hull $CH_B(A)$.

\end{algorithmic}
\end{algorithm}
\caption{Pseudocode of the new RCH algorithm (Main program).}
\label{CodeMain}
\end{figure}

%\newpage

\begin{figure}
\begin{algorithm}[H]
\begin{algorithmic}[1]

\Statex

\Procedure {CAV}{$CH(A)$, i, s, P1, P2}

\State Initialize local variables $u=0, v=0, w=0$.

\State Generate the polygon $O$ by all vertices of P1,
from $CH(A)_i$ up to $CH(A)_{i+1}$.

\State Form the polyline $I$ by $CH(A)_{i+1}$, $CH(A)_i$, and all vertices of $P2$ inside $O$ or collinear with its frontier with exception of the cover
$\overline{CH(A)_{i+1}CH(A)_i}$. $I$ has $N$ vertices.

\If {$N>2$}

\State Determine the list $CH(I)$ which has $S$ vertices.

\State Insert between $CH(A)_i$ and $CH(A)_{i+1}$ the vertices of $CH(I)$,
with exception of the first point and $CH(A)_i$, $CH(A)_{i+1}$.

\State u=S-3

\State s=s+u

\State Generate the polygon $O_S$ by $CH(A)_i$,
$CH(A)_{i+1}$ and all vertices of $P2$ previous to $CH(A)_{i+1}$
up to a first point found outside $O$.

\State Form the polygon $I_S$ by $CH(A)_{i+u+1}$, $CH(A)_{i+u}$ and all vertices of $P1$ inside
$O_S$ or collinear with its frontier (with exception of the line segment $\overline{CH(A)_{i+u+1}CH(A)_{i+u}}$). $I_S$ has $N_S$
vertices.

\If {$N_S>2$}

\State Determine the list $CH(I_S)$ which has $S$ elements.

\State Insert between $CH(A)_i$ and $CH(A)_{i+1}$ the elements of $CH(I_S)$ with exception of the first one and $CH(A)_{i+u+1}$ and $CH(A)_{i+u}$.

\State v=S-3

\State s=s+v

\State w=u+v

\EndIf

\State Generate the polygon $O_E$ by $CH(A)_{i+w+1}$,
$CH(A)_{i+w}$ and all vertices of $P2$ subsequent to
$CH(A)_{i+w}$ up to the first point found outside $O$.

\State Form the polygon $I_E$ by $CH(A)_{i+w+1}$,
$CH(A)_{i+w}$ and all vertices of $P1$ inside $O_E$ or
collinear with its frontier (with exception of the line segment
$\overline{CH(A)_{i+w+1}CH(A)_{i+w}}$). $I_E$ has $N_E$ vertices.

\If {$N_E>2$}

\State Determine the list $CH(I_E)$ which has $S$ elements.

\State Insert between $CH(A)_{i+w}$ and $CH(A)_{i+w+1}$ the elements of the list $CH(I_E)$ with exception of the first one and $CH(A)_{i+w}$ and $CH(A)_{i+w+1}$.

\State x=S-3

\State s=s+x

\EndIf

\State i=i-1

\EndIf

\State \Return {$CH(A), i, s$}

\EndProcedure

\end{algorithmic}
\end{algorithm}
\caption{Pseudocode of the new RCH algorithm (Cavity processing procedure).}
\label{CodeProcedure}
\end{figure}

\end{document}